\documentclass[11pt]{article}
\usepackage{times}  
\usepackage{helvet}  
\usepackage{courier}  
\usepackage{url}  
\usepackage{graphicx}  
\usepackage{amssymb}
\usepackage{amsthm}
\usepackage{amsmath}
\usepackage{enumerate}
\usepackage{mathtools}
\usepackage{amsfonts,amsthm,epsfig,epstopdf,url,array}
\usepackage{dsfont}
\usepackage{breqn}

\title{On the Complexity of the Inverse Semivalue Problem\\ 
for Weighted Voting Games\footnote{An extended abstract 
of this work appears in the Proceedings of the Thirty-Third AAAI 
Conference on Artificial Intelligence (AAAI 2019).}}

\author{
Ilias Diakonikolas\thanks{Supported by NSF Award CCF-1652862 (CAREER) and a Sloan Research Fellowship.}\\
Computer Science Department\\
University of Southern California, USA\\
{\tt diakonik@usc.edu}\\
\and
Chrystalla Pavlou\thanks{Supported by EPSRC Scholarship. Some of this work was performed while visiting USC.}\\
School of Informatics\\
University of Edinburgh, UK\\
{\tt C.Pavlou@sms.ed.ac.uk }
}

\theoremstyle{plain} \newtheorem{definition}{Definition}
\theoremstyle{plain} \newtheorem{theorima}{Theorem}
\theoremstyle{plain} \newtheorem{lemma}{Lemma}
\theoremstyle{plain} 
\theoremstyle{plain} 
\theoremstyle{plain} \newtheorem{prop}{Proposition}
\theoremstyle{plain} \newtheorem{rem}{Remark}
\begin{document}
\maketitle

\thispagestyle{empty}
\setcounter{page}{0}

\begin{abstract}
Weighted voting games are a family of cooperative games, 
typically used to model voting situations where a number of agents (players) vote against or for a proposal.  
In such games, a proposal is accepted if an appropriately weighted sum of the votes exceeds a prespecified threshold. 
As the influence of a player over the voting outcome is not in general proportional to her assigned weight, 
various power indices have been proposed to measure each player's influence. 
The inverse power index problem is the problem of designing a weighted 
voting game that achieves a set of target influences according to a predefined power index. 
In this work, we study the computational complexity of the inverse problem when 
the power index belongs to the class of semivalues. We prove that the inverse problem 
is computationally intractable for a broad family of semivalues, including all regular semivalues. 
As a special case of our general result, 
we establish computational hardness of the inverse problem for the Banzhaf indices and the Shapley values, 
arguably the most popular power indices. 
\end{abstract}

\newpage

\section{Introduction} \label{sec:intro}
\subsection{Background and Motivation}
Weighted voting games are a classical family of cooperative games that have been extensively studied 
in the game theory and social choice literature. Such games model a common voting scenario where each agent (player), 
associated with a weight, casts a ``YES'' (for) or ``NO'' (against) vote: if the weighted sum of the ``YES'' votes exceeds  
a threshold, then the voting outcome is ``YES'', otherwise the outcome is ``NO''. Examples of such practical scenarios 
include the voting system of the European Union, stockholder companies, and resource allocation in multi-agent systems \cite{EGGW08,KKZ11}.

Although having a larger weight might help an agent affect the voting outcome, 
her influence on the result of the game is not always proportional to her weight. 
Thus, instead of using agents' weights, the power of an agent over the outcome 
is usually measured in a systematic way by a {\em power index}. 
Over the years, many power indices have been proposed and studied, 
such as the Shapley value~\cite{shapley} (also known as Shapley-Shubik index 
for weighted voting games~\cite{SS}), the Banzhaf index~\cite{ban}, 
the Deegan-Packel index~\cite{DP}, and the Holler index~\cite{Hol}. 
The problem of computing the agents' power indices in a given game 
has received ample attention and its computational complexity 
is well-understood for many game representations and power index functions 
(see, e.g., ~\cite{PK90,DP94,A08}). 

\subsection{Our Contributions} 
In this work, we focus on the {\em inverse power index problem} 
--- that is, the problem of {\em designing} a weighted voting game  with a given set of power indices.
As we will explain in detail below, the inverse problem has been extensively studied in various fields, 
including game theory, social choice theory, and learning theory. Various works have provided 
heuristic methods, exponential time algorithms, or polynomial time approximation algorithms
with provable performance guarantees for this problem. 

Despite this wealth of prior work on the algorithmic version of the inverse problem, 
prior to this work its computational complexity was not well-understood, even for the most popular
power indices (Shapley values, Banzhaf indices). {\em In this paper, we study and essentially resolve the computational complexity of the inverse power index problem for weighted voting games, with respect to a broad and extensively studied 
family of power indices. Specifically, we show that the inverse problem is computationally intractable 
under standard complexity assumptions.} 
More specifically, we prove that for a large class of power indices --- that includes the popular Banzhaf index, Shapley values, 
and the class of semivalues~\cite{W79} --- the inverse problem cannot be in the polynomial hierarchy (PH), 
unless the polynomial hierarchy collapses. Prior to this work, it was conceivable that there 
exists an exact polynomial time algorithm for this problem. It follows from our hardness result 
that the existence of such an algorithm is unlikely.

\subsection{Related work} 
Several heuristic algorithms for the inverse Banzhaf index problem have been proposed 
in the social choice and game theory literature. 
Aziz, Paterson, and
Leech \cite{APL07} give an approximation algorithm that on input target Banzhaf indices and a desired $\ell_2$-distance bound
outputs a weighted voting game with integer weights that has Banzhaf indices within the desired distance bound. 
Unfortunately, no theoretical guarantees are provided regarding the convergence rate of this method, 
and it is not known whether it converges to an approximately optimal solution. Two related heuristic algorithms 
that return weighted voting games with Banzhaf indices within a given distance from the target indices are proposed 
by  Laruelle and Widgren~\cite{LW98} and Leech~\cite{L03}. 
Similarly to~\cite{APL07}, Fatima, Wooldridge, and Jennings  \cite{FWJ08} 
give an iterative approximation algorithm for the inverse Shapley value problem that on input target Shapley values, 
a quota, and a desired average percentage difference, outputs a weighted voting game with the given quota 
that has Shapley values within the desired distance. 
It is shown that each iteration runs in quadratic time and that the algorithm eventually converges, 
but no theoretical guarantees are given regarding the convergence rate, i.e., the time 
until a desired approximation is achieved.

An exact algorithm for both the inverse Banzhaf index and inverse Shapley value problems is given by Kurz~\cite{K11}.
The proposed method relies on integer linear programming and returns a weighted voting game 
that minimizes the $\ell_1$-distance from the target power indices. 
This method has running time exponential in the number of players. 
Another exact, but exponential-time, algorithm for the inverse power index problem 
is given by De Keijzer, Klos, and Zhang~\cite{KKZ11}. 
The~\cite{KKZ11} algorithm outputs a weighted voting game 
in which the power indices of the players are as close to the target vector as possible. 
The presented algorithm is based on an enumeration of all weighted voting games: 
for each weighted voting game, the algorithm computes the power indices, 
their distance from the target ones and it then outputs the game with the smallest distance. 
Since there exist $2^{\Omega(n^2)}$ weighted voting games with $n$ players, 
this algorithm also runs in exponential time.

In addition to the several heuristics and exponential time algorithms that have been proposed, 
a line of recent works in theoretical computer science~\cite{OS11, DDFS12,DDS12} have obtained 
polynomial time approximation algorithms with provable performance guarantees
for the inverse problem with respect to both the Banzhaf indices and the Shapley values. 
These algorithms output a weighted voting game whose power indices have 
small $\ell_2$-distance from the target indices. These algorithmic results were recently 
extended to more general classes of functions~\cite{DK18-degd}.


The inverse power index problem is also of significant interest in various other fields, 
such as circuit complexity and computational learning theory. The reader is referred to~\cite{OS11}
for a detailed summary of this connection. In the fields of computational complexity and learning theory, 
linear threshold functions (LTFs), 
which are equivalent to weighted voting games if we allow negative weights~\cite{KKZ11}, 
have been of great significance and have been studied for several decades~\cite{Rosenblatt:58, chow, 
MTT:61, Dertouzos:65, MinskyPapert:68}.
A fundamental result of C. K. Chow from the the early 1960s~\cite{chow} 
shows that linear threshold functions are characterized by their degree-$0$ and degree-$1$ ``Fourier coefficients'', 
now known as {\em Chow parameters}. Given this structural result, the following natural computational question
--- now known as ``the Chow parameters problem''~\cite{OS11}  --- arises: 
Given the Chow parameters of an LTF, 
reconstruct a weights-based representation of the function. Interestingly, the Chow parameters 
are essentially equivalent to the non-normalized Banzhaf indices~\cite{DPS79}, and therefore 
the inverse Banzhaf index problem is tantamount to the Chow parameters problem.

In addition to the aforementioned algorithmic results, a number of 
complexity results have been established concerning weighted voting games.
Aziz \cite{A08}, Elkind {\em et al.}~\cite{EGGW07}, and Elkind {\em et al.}~\cite{EGGW08} 
study the computational complexity of various problems related to weighted voting games. 
Aziz~\cite{A08} studies the complexity of computing various indices such as the Shapley values, 
the Banzhaf, and the Deegan-Packel indices for a given simple game when the game is given in different forms.  
Finally, Faliszewski and Hemaspaandra \cite{FH08} study the complexity of the power index {\em comparison problem}: 
given two weighted voting games and a player, decide on which game the given player 
has higher influence as it is computed by a specific power index. They show that this problem is intractable, 
namely  PP-complete, for both the Shapley values and Banzhaf index. To achieve this, they extend the 
\#P-metric-completeness of computing the Shapley values, proved by Deng and Papadimitriou~\cite{DP94}. 
They prove that, whereas computing the Banzhaf indices of a weighted voting game is \#P-parsimonious-complete \cite{PK90}, computing the Shapley values is \#P-many-one complete and it cannot be strengthened to \#P-parsimonious-complete.
Gopalan, Nisan, and Roughgarden \cite{GNR15} study the convex polytope 
consisting of the Chow parameters of all Boolean functions. 
They show that the linear optimization problem over this polytope is \#P-hard; 
a result that indicates, but does not logically imply, that the inverse Banzhaf index problem may be intractable.


\section{Preliminaries}    
\paragraph{Notation}
We write $wt(x)$ to denote the {\em weight} of a Boolean vector $x \in \{-1,1\}^n$, 
i.e., the number of 1's in $x$. For any $i \in \{1, \dots, n\}$ and $x \in \{-1,1\}^n$ such that $x_i=-1$, 
we write $x^i$ to denote the vector obtained when we flip the $i^{th}$ coordinate of $x$. 
We denote by $\textbf{1}$ (resp. $\textbf{-1}$) the vector in $\{-1, 1\}^n$ with all coordinates equal to $1$ (resp. $-1$).
We will denote by $\mathrm{Conv}(S)$ the convex hull of the set $S$.
We will use $\mathrm{sign}: \mathbb{R} \rightarrow \mathbb{R}$ for the function 
that takes value $1$ if $z \geq 0$ and value $-1$ if $z<0$.

\smallskip

\noindent Our basic object of study is the family of linear threshold functions (LTFs) over $\{-1,1\}^n$:

\begin{definition}[Linear Threshold Function] \label{def:ltf}
A {\em linear threshold function (LTF)} is any function 
$f_{w, \theta}: \{-1,1\}^n \rightarrow \{-1, 1\}$ such that 
$f_{w, \theta}(x)=\mathrm{sign}(w \cdot x - \theta)$ for some weight vector 
$w \in \mathbb{R}^n$ and  threshold $\theta \in \mathbb{R}$.
\end{definition}

Note that weighted voting games are equivalent to LTFs with {\em non-negative weights}.
We leverage this equivalence throughout this paper. At various points, we 
may refer to a weighted voting game as an LTF without further elaboration.


\paragraph{Semivalues}
We mainly focus on power indices that belong to {\em the class of semivalues}. 
Semivalues are a fundamental family of power indices, introduced by Weber \cite{W79} 
as generalizations of the Shapley value that do not satisfy the efficiency axiom \cite{DNW81}. 
Since their introduction, semivalues have received considerable attention, see, e.g., \cite{E87,CFP03,CF08}.

We start by providing the definition of semivalues \cite{roth} in terms of weighting coefficients, 
as they were characterized by Dubey, Neyman, and Weber \cite{DNW81}:
\begin{definition}[Semivalues] \label{def:semiv}
For a positive integer $n$, a probability vector $p^n=(p_0^n, \dots, p_{n-1}^n) \in \mathbb{R}^n$ is 
a vector such that $\sum_{t=0}^{n-1}\binom{n-1}{t}p_{t}^n=1$ and $p_t^n \geq 0,\text{ for } t \in \{0, \dots, n-1\}$. 
The $i$-th semivalue corresponding to the probability vector $p^n$ of a Boolean function 
$f:\{-1,1\}^n \rightarrow \{-1,1\}$ is defined to be
\begin{align*}
\widetilde{f}^{p^n}(i)=\sum_{x \in \{-1,1\}^n: x_i=-1}p^n_{wt(x)}f(x)x_i+\sum_{x \in \{-1,1\}^n: x_i=1}p^n_{wt(x)-1}f(x)x_i \;,
\end{align*}
for $i \in \{1, \dots, n\}$. 
\end{definition}

Intuitively, we can interpret $p^n$ as the vector of probabilities that a given player 
will join a coalition of size $t$~\cite{CF08}, $0 \leq t \leq n-1$. With this interpretation, 
the $i$-th semivalue computes the probability of the event that player $i$ is a {\em pivot}, 
i.e., the probability that the output of the game would change from $1$ to $-1$ 
if the $i$-th player (the $i$-th variable) were to change her vote from $1$ to $-1$.

We call a semivalue {\em regular} if it is defined by strictly 
positive probability vectors~\cite{CF08}.

\paragraph{Remark}
The Shapley values and Banzhaf indices are the semivalues defined 
by $p^n_t=\frac{(n-t-1)! t!}{n!}$ \cite{SS} and $p^n_t=\frac{1}{2^{n-1}}$ \cite{DPS79}, respectively. 
Note that both indices are regular semivalues.

\paragraph{Reformulation of Semivalues}
The following equivalent way to express a set of semivalues will 
be useful throughout this paper.
Setting $p_{-1}^n=p_{n}^n=0$, we observe that we can rewrite the semivalues vector defined by the probability vector 
$p^n$ as follows.
For $i \in \{1, \dots, n\}$, we have:
\begin{align}\label{sem}
\widetilde{f}^{p^n}(i)&= \frac{1}{2}\sum_{x \in \{-1,1\}^n }f(x)x_i \left(p^n_{wt(x)}+p^n_{wt(x)-1}\right)
+\frac{1}{2}\sum_{x \in \{-1,1\}^n}f(x) \left(p^n_{wt(x)-1}-p^n_{wt(x)}\right) \;.
\end{align}
From this representation, a probability distribution $\mu_{p^n}$ over $\{-1,1\}^n$ emerges,
defined as follows. For $x \in \{-1,1\}^n$, we have: 
$$\mu_{p^n}(x): = \frac{\mu'_{p^n}(x)}{\Lambda(p^n)} \;,$$
where $\mu'_{p^n}(x): =p^n_{wt(x)}+p^n_{wt(x)-1}$ and 
$\Lambda(p^n):=\sum_{x \in \{-1,1\}^n}\mu'_{p^n}(x)$ is the normalizing factor. 

For a Boolean function $f: \{-1,1\}^n \rightarrow \{-1, 1\}$, we will write 
$\widehat{f}^{p^n}(i): =\sum_{x \in \{-1,1\}^n} \mu'_{p^n}(x)f(x)x_i$ for the first term of (\ref{sem}) 
and $C_f^{p^n} := \sum_{x \in \{-1,1\}^n}f(x)(p^n_{wt(x)-1}-p^n_{wt(x)})$ for the second term of (\ref{sem}). 
One can view the first term, $\widehat{f}^{p^n}(i)$, as the expectation $\mathbb{E}_{x \sim  \mu_{p^n}} [f(x)x_i]$ 
(up to the normalizing factor $\Lambda(p^n)$).

\medskip

We now define the notion of a {\em reasonable} probability vector to describe
the family of semivalues for which our computational hardness results apply:

\begin{definition}[Reasonable Probability Vector] \label{def:reas}
A probability vector $p^n \in \mathbb{R}^{n}$ is called {\em reasonable} if there exists a $t$ with
$t = \Omega(n)$ and $n-t = \Omega(n)$ such that $p^n_t > 0$. 
\end{definition}

The intuition behind the above definition is that the distribution $\mu_{p^n}$
has support $2^{\Omega(n)}$. Note that this happens when there exists
a $t$ with $\binom{n}{t} = 2^{\Omega(n)}$ such that $p^n_t > 0$. 
We recall that all regular semivalues (including the Banzhaf indices and Shapley values)
satisfy this property.

\begin{rem}
{\em For computational purposes, throughout this paper, we will assume that 
each value $p^n_t$ defining our probability vector is a rational number 
that can be described as a ratio of integers with $\mathrm{poly}(n)$ bits.}
\end{rem}

\paragraph{Inverse Semivalues Problem} We are ready to define the inverse semivalues problem. 
Let $n$ denote the number of players and consider the semivalues defined by a known
probability vector $p^n$. Given a vector $c = (c_1, \ldots, c_n)$ of target semivalues, 
we want to either find a weighted voting game with these target
semivalues or decide that there does not exist any weighted voting game with 
semivalues vector $c$. 

\smallskip

\medskip
\fbox{\parbox{6in}{
\noindent \textbf{Name:} SV$_{p^{n}}$-Inverse Problem

\medskip

\textbf{Input:} A vector $(c_1, \dots, c_{n}) \in \mathbb{Q}^{n}$ and $\theta \in \mathbb{Q}$.

\smallskip 

\textbf{Question:} Output $w \in \mathbb{Q}_{+}^n$ with $\sum_{i=1}^n w_i=1$ such that $\widetilde{f}_{w, \theta}^{p^n}(i)=c_i$, 
for $i \in \{1,\dots, n\}$, or ``NO'' if no such $w$ exists. 

}}

\section{Main Result: Computational Intractability of Inverse Power-Index Problem} \label{sec:main-thm}


\subsection{Statement of Main Result and Proof Overview}

The main result of this paper is the following:

\begin{theorima}[Main Result, Informal Statement] \label{thm:main}\label{main}
For semivalues defined by the probability vector $p^n$, if $p^n$ is a reasonable probability vector, 
then the SV$_{p^n}$-Inverse problem is not in the polynomial hierarchy, unless the polynomial hierarchy collapses.
\end{theorima}

As an immediate corollary of Theorem~\ref{thm:main}, we obtain that 
the inverse power index problem is similarly intractable for the class of regular semivalues, 
which includes the Shapley values and the Banzhaf indices.

\paragraph{Proof Overview}
We start with a brief overview of our proof establishing Theorem~\ref{main}.
To prove hardness of the inverse problem, we examine the convex polytope $\mathcal{C}_{p^{n+2}}$ consisting of the convex combinations of the semivalues of linear threshold functions with zero threshold and weight vectors 
of a specific form described below (Definition~\ref{poly}). 
We prove (Theorem~\ref{thm:opt-hard}) that if the probability vector $p^n$ defining 
the semivalues is {\em reasonable}, then the  linear optimization over $\mathcal{C}_{p^{n+2}}$ 
is \#P-hard (under Turing reductions). Then, we proceed to show (Theorem~\ref{thm:svr-hard}) 
that the optimization problem can be solved using an oracle for the semivalues {\em verification problem},  
i.e., the problem of verifying that the given target semivalues are the actual semivalues of a given linear threshold function, 
or an oracle for the {\em inverse problem} for weight vectors of the aforementioned specific form. 
We thus conclude that the verification and the inverse semivalues problem for linear threshold functions with this 
specific weight structure cannot be in the polynomial hierarchy. Finally, using a lemma that shows 
that semivalues characterize the space of linear threshold functions with the same threshold (Lemma~\ref{equiv}), 
we show hardness of the inverse and verification problems for linear threshold functions with positive weights, 
i.e., for weighted voting games, as desired.  

Our proof strategy bears some similarities to the approach by Gopalan, Nisan, and Roughgarden \cite{GNR15} 
(also exploited by Dughmi and Xu \cite{DH16}). In particular, Gopalan, Nisan, and Roughgarden \cite{GNR15} show 
that the linear optimization problem over the polytope consisting of the Chow parameters of all Boolean functions is \#P-complete,
and therefore there cannot exist an efficient membership oracle for this polytope.  We note that our results include the results of \cite{GNR15} regarding Chow parameters as a very special case: as previously mentioned, the non-normalized Banzhaf indices of a linear threshold function are equal to its Chow parameters and they are semivalues.  

Despite this similarity, our proof involves a number of novel ideas that seem necessary 
in order to handle a broad range of probability distributions that could define a semivalue. 
One of the difficulties comes from the fact that we want to prove hardness for the class 
of weighted voting games, i.e., LTFs with {\em positive} weights. While this requirement is easy 
to handle for the Banzhaf indices (Chow parameters), it poses non-trivial difficulties for more 
general semivalues. To handle this, we propose a generalization of the definition of the Khintchine constant 
from the uniform distribution to any probability distribution and establish that it is hard to compute 
under a restricted set of weights that is crucial for our proof (Theorem~\ref{thm:KK-hard}). 
Another crucial ingredient of our proof is a new structural result (Lemma~\ref{equiv}) 
establishing that the set of semivalues uniquely determines a weighted voting game. 

\subsection{Proof of Main Result}
In this subsection, we proceed with the detailed proof of 
Theorem~\ref{thm:main}.

\paragraph{Semivalues Polytope} Our analysis makes essential use of the 
convex polytope defined as the convex hull of the set of semivalues for 
all linear threshold functions whose weights-based representation is of a specific form: 
Namely, their threshold $\theta=0$ and their weight vectors consist of $n$ positive coordinates 
and two coordinates each of whose weights is equal to minus a half times the sum of the first $n$ coordinates.
Formally, we introduce the following definition:

\begin{definition} \label{poly}
For a positive integer $n$, define $\mathcal{C}_{p^{n+2}} := \mathrm{Conv}(\mathcal{A}_{p^{n+2}})$, 
where 
$$\mathcal{A}_{p^{n+2}} := \left\{c \in \mathbb{R}^{n+2}: \exists w \in \mathbb{R}^{n+2}, \; w_i > 0, \;1 \leq i \leq n,\; w_{n+1}=w_{n+2}=-\sum_{i=1}^nw_i/2,\; c_i=\widetilde{f}^{p^{n+2}}_{w,0}(i), \; 1 \leq i \leq n+2 \right\} \;.$$
\end{definition}

The first main step of our proof involves showing that the linear optimization problem
over the above defined polytope is computationally hard.

\paragraph{Linear Optimization over $\mathcal{C}_{p^{n+2}}$}
We firstly prove that if the semivalues' probability distribution defined by $p^{n+2}$ has a sufficiently large  support
(in particular, if $p^{n+2}$ is a reasonable vector as in Definition~\ref{def:reas}), 
then the linear optimization problem over the polytope $\mathcal{C}_{p^{n+2}}$ is \#P-hard. 

The linear optimization problem for semivalues defined by any probability vector $p^{n+2}$ 
is captured by the following family of problems:

\medskip

\fbox{\parbox{6in}{

\noindent \textbf{Name:} SV$_{p^{n+2}}$-Optimization Problem 

\smallskip

\textbf{Input:} A vector $a=(a_1, \dots, a_{n+2}) \in \mathbb{Q}^{n+2}$.\\
\textbf{Question:} Compute $\max_{c \in \mathcal{C}_{p^{n+2}}} a \cdot c$. 

}}

\medskip

The main result of this subsection is the following:

\begin{theorima} \label{thm:opt-hard}
If $p^{n+2}$ is a {\em reasonable} probability vector, the SV$_{p^{n+2}}$-Optimization Problem is \#P-hard.
\end{theorima}

\smallskip

We prove Theorem~\ref{thm:opt-hard} by reducing from an intermediate problem --- that of 
computing the Khintchine constant of a vector {\em with respect to the probability distribution $\mu_{p^n}$}:

\medskip

\fbox{\parbox{6in}{
\noindent \textbf{Name:} Khintchine$_{\mu_{p^n}}$ 

\smallskip 

\textbf{Input:} A vector $a = (a_1, \dots, a_n) \in \mathbb{Q}^{n}$.\\
\textbf{Question:} Compute $K_{\mu_{p^n}}(a)=\mathbb{E}_{x \sim \mu_{p^n}}[|a \cdot x|]$.  

}}

\medskip

The Khintchine constant has been extensively studied with respect to the uniform distribution
on the Boolean hypercube (see, e.g.,~\cite{Szarek:76, DeDS16} and references therein). We note that 
\cite{GNR15} established the intractability of computing this quantity under the uniform distribution.
We show:

\begin{theorima} \label{thm:KK-hard}
If $p^{n+2}$ is a reasonable probability vector, the Khintchine$_{\mu_{p^{n+2}}}$ problem is \#P-hard, 
even restricted to inputs $(a_1, \dots, a_n, -A/2, -A/2)$, where $A=\sum_{i=1}^n a_i$ 
and $a_i>0$ for $i \in \{1, \dots, n\}$.
\end{theorima}
\begin{proof}

We start by showing that the \#Partition problem for the distribution $\mu_{p^n}$ is hard 
and then reduce the latter problem to the former.
 
\medskip

\noindent \textbf{Name:} \#Partition$_{\mu_{p^n}}$

\smallskip

\noindent \textbf{Input:} A vector $w=(w_1, \dots, w_n) \in \mathbb{Z}^{n}$.\\
\textbf{Question:} Compute $\text{Pr}_{x \sim \mu_{p^n}}[w \cdot x = 0]$.  

\medskip 

We start with the following proposition:

\begin{prop} \label{thm:partition}
If $p^{n+2}$ is a reasonable probability vector, \#Partition$_{\mu_{p^{n+2}}}$ is \#P-hard, 
even restricted to inputs $(w_1, \dots, w_n, -W/2, -W/2)$, 
where $W=\sum_{i=1}^n w_i$ and $w_i>0$ for $1 \leq i \leq n$.
\end{prop}
\begin{proof} We reduce from the following problem:

\medskip

\noindent \textbf{Name:} \#R-Partition \\ 
\textbf{Input:} Positive integers $c_j,1 \leq j \leq n$ and a positive integer $b_1n \leq k \leq b_2n$, 
where $0 < b_1 \leq b_2 < 1$, such that if there is a subset $S$ of $\{1,\dots,n\}$ 
with $\sum_{i \in S} c_i=\sum_{i=1}^{n} c_i/2$, then $|S|=k$ or $|S|=n-k$.\\ 
\textbf{Question:} Compute the number of subsets $S$ of $\{1,\dots,n\}$ such that 
$\sum_{i \in S} c_i=\sum_{i=1}^{n} c_i/2$. 

\medskip

Given an instance $c_1, \dots, c_{n}$ of the \#P-complete \#R-Partition \cite{DP94}, 
we construct an instance of \#Partition$_{\mu_{p^{n+2}}}$ as follows: 
We set $w_i = c_i$, for $1 \leq i \leq n$, and define $w= (w_1, \dots, w_{n}, -W/2, -W/2)$, 
where $W=\sum_{i=1}^n w_i$. We have: 

\begin{align*} 
\text{Pr}_{x \sim \mu_{p^{n+2}}} [w \cdot x =0]
&=\sum_{x \in \{-1,1\}^{n+2}: w \cdot x =0}\mu_{p^{n+2}}(x)=
\frac{|\{x:w \cdot x=0, x \not \in \{\textbf{-1},\textbf{1}\}\}|}{2\Lambda(p^{n+2})}(p^{n+2}_{n-k+1}+p^{n+2}_{n-k})\\
&+\frac{|\{x:w \cdot x=0, x \not \in \{\textbf{-1},\textbf{1}\}\}|}{2\Lambda(p^{n+2})}(p^{n+2}_{k+1}+p^{n+2}_{k})
+\frac{1}{\Lambda(p^{n+2})}(p^{n+2}_{n+1}+p^{n+2}_{0}) \;,
\end{align*} 
as the \#R-Partition problem guarantees that every solution has size $k$ or $n-k$ 
and the number of solutions with size $k$ are equal to the number of solutions with size $n-k$. 
That is, every $x \in \{-1,1\}^{n+2} \setminus \{\textbf{-1},\textbf{1}\}$ such that $w \cdot x=0$ 
is guaranteed to have weight $k+1$ or weight $n-k+1$ as $x_{n+1}$ has to be different 
than $x_{n+2}$, otherwise the only solutions are $\{\textbf{-1},\textbf{1}\}$. 
So, for every solution of the \#R-partition problem we have two 
$x \in \{-1,1\}^{n+2} \setminus \{\textbf{1}, \textbf{-1}\}$ such that $w \cdot x=0$. 
Thus, 
\begin{align*}
\left| \left\{ S \subset \{1,\dots,n\}:\sum_{i \in S} c_i =\sum_{i=1}^{n}c_i/2 \right\} \right| =  
\frac{\Lambda(p^{n+2})\text{Pr}_{x \sim \mu_{p^{n+2}}} [w \cdot x =0]-(p^{n+2}_{n+1}+p^{n+2}_{0})}{(p^{n+2}_{n-k+1}+p^{n+2}_{n-k}+p^{n+2}_{k+1}+p^{n+2}_{k})} \;. 
\end{align*} 
This completes the proof of Proposition~\ref{thm:partition}.
\end{proof}

Given an instance of \#Partition$_{\mu_{p^{n+2}}}$, i.e., 
$a=(a_1,\dots,a_{n}, -A/2, -A/2)$, where $A=\sum_{i=1}^na_i$ and $a_i > 0$ for $1 \leq i \leq n$, 
we construct the following three Khintchine$_{\mu_{p^{n+2}}}$ instances:
\begin{align*}
    c&=2(a_1, a_2,\dots, a_{n},-A/2,-A/2), \\ 
    d&=(a_1-y, a_2, \dots, a_{n},-A/2+y/2,-A/2+y/2),\\
    e&=(a_1+y,a_2, \dots, a_{n},-A/2-y/2,-A/2-y/2),
\end{align*}
where $0<y< 1/2 < a_1$. We will show that solving the above instances of Khintchine$_{\mu_{p^{n+2}}}$ suffices
to solve our given instance of \#Partition$_{\mu_{p^{n+2}}}$. This will complete the proof of Theorem~\ref{thm:KK-hard}. 
To do so, we require some case analysis and explicit calculations. 

\smallskip


\noindent For any $x \in \{-1,1\}^{n+2}$, we have that:
\begin{align*}
 &|d \cdot x|= \begin{cases} 
       |\sum_{i=1}^na_ix_i +A-2y|, & \textrm{ if  } x_1=1, x_{n+1}=x_{n+2}=-1\\
	  |\sum_{i=1}^na_ix_i+A|, & \textrm{ if  } x_1=-1, x_{n+1}=x_{n+2}=-1\\
      |\sum_{i=1}^na_ix_i-A|, &\textrm{ if  } x_1=1,x_{n+1}=x_{n+2}=1 \\
     |\sum_{i=1}^na_ix_i-A+2y|, &\textrm{ if  } x_1=-1, x_{n+1}=x_{n+2}=1\\ 
  |\sum_{i=1}^na_ix_i-y|, &\textrm{ if  } x_1=1,x_{n+1}\neq x_{n+2}\\ 
|\sum_{i=1}^na_ix_i+y|, &\textrm{ if  } x_1=-1,x_{n+1}\neq x_{n+2}\\  
   \end{cases} 
\end{align*}
and similarly
\begin{align*}
 |e \cdot x|= \begin{cases} 
       |\sum_{i=1}^na_ix_i+A+2y|, &\textrm{ if  } x_1=1,x_{n+1}=x_{n+2}=-1\\
	  |\sum_{i=1}^na_ix_i+A|, &\textrm{ if  } x_1=-1, x_{n+1}=x_{n+2}=-1\\
      |\sum_{i=1}^na_ix_i-A|, &\textrm{ if  } x_1=1, x_{n+1}=x_{n+2}=1 \\
     |\sum_{i=1}^na_ix_i-A-2y|, &\textrm{ if  } x_1=-1, x_{n+1}=x_{n+2}=1\\ 
  |\sum_{i=1}^na_ix_i+y|, &\textrm{ if  } x_1=1, x_{n+1}\neq x_{n+2}\\ 
|\sum_{i=1}^na_ix_i-y|, &\textrm{ if  } x_1=-1,x_{n+1}\neq x_{n+2}\\ 
   \end{cases} 
\end{align*}
Given the above, we observe that the following hold: 
\begin{itemize}
\item For $x \in \{-1,1\}^{n+2}$, $wt(x) \in \{0, n+2\}$,
\[|c \cdot x|= |e \cdot x|=|d \cdot x|=0 \;.\]
\item For $x \in \{-1,1\}^{n+2}$, $x_{n+1} = x_{n+2}=x_1$,
\[|d \cdot x| +|e \cdot x|=|c \cdot x| \;.\]
\item For $x \in \{-1,1\}^{n+2}$, $x_{n+1} \neq x_{n+2}$: 
\begin{itemize}
\item If $c \cdot x \neq 0$,
\[|d \cdot x|+|e \cdot x|=|c \cdot x| \;,\]
as it holds that 
$|\sum_{i=1}^n a_ix_i-y|+|\sum_{i=1}^n a_i x_i+y|=2\max(|\sum_{i=1}^n a_i x_i|,|y|)$ 
and $|\sum_{i=1}^n a_i x_i| \geq 1$.
\item If $c \cdot x=0$, then
\[ |d \cdot x|+|e \cdot x|=2y \;. \] 
\end{itemize}
\item For $x \in \{-1,1\}^{n+2}$, $x_{n+1}=x_{n+2}$, $x_1 \neq x_{n+1}$
\[|d \cdot x|+|e \cdot x|=|c \cdot x| \;, \]
as similarly with the above case, we have that 
$|\sum_{i=1}^{n+2}a_ix_i-2y|+|\sum_{i=1}^{n+2}a_i x_i+2y|=2\max(|\sum_{i=1}^{n+2}a_i x_i|,|2y|)$ 
and $|\sum_{i=1}^{n+2} a_i x_i| \geq 1 $.
\end{itemize}
Hence, we have:  
\begin{dmath*}
    K_{\mu_{p^{n+2}}}(d)+K_{\mu_{p^{n+2}}}(e)-K_{\mu_{p^{n+2}}}(c) = 
    \mathbb{E}_{x \sim \mu_{p^{n+2}}}\left[ |d \cdot x| \right]+
    \mathbb{E}_{x \sim \mu_{p^{n+2}}}\left[|e \cdot x|\right]-
     \mathbb{E}_{x \sim \mu_{p^{n+2}}}\left[|c \cdot x|\right] = 
    \sum_{x \in \{-1,1\}^{n+2}} 2y \cdot \mu_{p^{n+2}}(x) \cdot \mathds{1}_{c \cdot x=0, x \not \in \{\bf{-1},\bf{1}\}} \;.
\end{dmath*}
So, we get:
\begin{align*}
\mathrm{Pr}_{x \sim \mu_{p^{n+2}}}[a \cdot x=0]=& \frac{K_{\mu_{p^{n+2}}}(d)+K_{\mu_{p^{n+2}}}(e)-K_{\mu_{p^{n+2}}}(c)}{2y}+\mathrm{Pr}_{x \sim \mu_{p^{n+2}}}[x=\textbf{-1}]+\mathrm{Pr}_{x \sim \mu_{p^{n+2}}}[x=\bf{1}]\\&=\frac{K_{\mu_{p^{n+2}}}(d)+K_{\mu_{p^{n+2}}}(e)-K_{\mu_{p^{n+2}}}(c)}{2y}+\frac{p_0+p_{n+1}}{\Lambda(p^{n+2})} \;.
\end{align*}
The proof of Theorem~\ref{thm:KK-hard} is now complete.
\end{proof}

We are now ready to prove Theorem~\ref{thm:opt-hard}.
\begin{proof}[Proof of Theorem~\ref{thm:opt-hard}]
We reduce from the Khintchine$_{\mu_{p^{n+2}}}$ problem: given the vector 
$a=(a_1, \dots, a_n, a_{n+1}=-A/2, a_{n+2}= -A/2) \in \mathbb{Q}^{n+2}$, 
where $A=\sum_{i=1}^n a_i$ and $a_i >0$ for $0 \leq i \leq n$, we want to compute 
$\max_{c \in \mathcal{C}_{p^{n+2}}}a \cdot c$. 
For any $c \in \mathcal{C}_{p^{n+2}}$, using our reformulation of semivalues (\ref{sem}), we have:
\begin{align*}
a \cdot c&= \frac{1}{2}\sum_{i=1}^{n+2}a_i \widehat{f}^{p^{n+2}}(i)+\frac{1}{2}\sum_{i=1}^{n+2}a_iC_f^{p^{n+2}}
\\&=\frac{\Lambda(p^{n+2})}{2}\sum_{x \in \{-1,1\}^{n+2}}f(x)\mu_{p^{n+2}}(x)\sum_{i=1}^{n+2}a_ix_i 
 \\&\leq \frac{\Lambda(p^{n+2})}{2}\sum_{x \in \{-1,1\}^{n+2}}\mu_{p^{n+2}}(x) \left| \sum_{i=1}^{n+2}a_ix_i \right|\;,
\end{align*}
where $f$ is a linear threshold function. 

Reducing from a weight vector that has the sum of its coordinates equal to zero was essential for this step: 
the term $C_f^{p^{n+2}}$ vanishes and we end up with the term 
$\sum_{x \in \{-1,1\}^{n+2}}f(x)\mu_{p^{n+2}}(x)\sum_{i=1}^{n+2}a_ix_i$ 
that is upper bounded by  $\sum_{x \in \{-1,1\}^{n+2}}\mu_{p^{n+2}}(x)|\sum_{i=1}^{n+2}a_ix_i|$. 

This upper bound is tight, as we show below, which is crucial for our argument.
If one were to reduce from a vector with sum different than $0$ or include in the polytope only linear threshold functions with positive weights, 
it would not have been possible to obtain a tight upper bound.

Observe that 
\[a \cdot c =\frac{\Lambda(p^{n+2})}{2} \sum_{x \in \{-1,1\}^{n+2}}\mu_{p^{n+2}}(x)\left|\sum_{i=1}^{n+2}a_ix_i \right|\;,\] 
for $c=(\widetilde{f}^{p^{n+2}}(1),\dots, \widetilde{f}^{p^{n+2}}(n+2))$, 
where $f(x)=\mathrm{sign}(\sum_{i=1}^{n+2}a_ix_i).$
So, 
\[\mathbb{E}_{x \sim \mu_{p^{n+2}}}\left[\left|\sum_{i=1}^{n+2}a_ix_i \right|\right]=
\frac{2\max_{c \in \mathcal{C}_{p^{n+2}}}a \cdot c}{\Lambda(p^{n+2})} \;.\]
That is, a solution to our instance of linear optimization over our polytope
gives a solution to the initial instance of the Khintchine$_{\mu_{p^{n+2}}}$ problem, 
and the proof of Theorem~\ref{thm:opt-hard} is complete.
\end{proof}

\paragraph{Linear Optimization Using a Verification Oracle} 
We prove that the {\em restricted verification problem}, defined below, is computationally hard, where 
the input linear threshold functions are defined by weight vectors of the specific 
form described in Definition \ref{poly}. 

\medskip

\fbox{\parbox{6in}{
\noindent \textbf{Name:} SVR$_{p^{n+2}}$-Verification Problem 

\smallskip

\textbf{Input:} A vector $(c_1, \dots, c_{n+2}) \in \mathbb{Q}^{n+2}$ and a vector $w=(w_1, \dots, w_{n+2})\in \mathbb{Q}^{n+2}$ such that $w_{n+1}=w_{n+2}=-\sum_{i=1}^nw_i/2$ and $w_i>0, i \in \{1, \dots, n\}$.\\
\textbf{Question:} Does it hold that $\widetilde{f}^{p^{n+2}}_{w, 0}(i)=c_i$ for $i \in \{1,\dots, {n+2}\}$?

}}

\medskip

\begin{theorima} \label{thm:svr-hard} \label{hi}
If $p^{n+2}$ is a reasonable probability vector, the SVR$_{p^{n+2}}$-Verification problem 
is not in the $k$-th level of the polynomial hierarchy, 
unless \#P is contained in the $(k+2)$-level.
\end{theorima}

The main idea behind the proof is that one can solve the linear optimization problem using a membership oracle 
of the polytope which can be obtained if we have an efficient algorithm for the restricted verification problem. 
Since the vertices of the polytope correspond to semivalues of linear threshold functions, if we have an efficient 
algorithm for the verification problem, then we can efficiently verify that a vector is a vertex of the polytope, 
and. using Caratheodory's theorem, we can get a membership oracle. In this way, we obtain a contradiction, 
unless the polynomial hierarchy collapses: if the verification problem is in the polynomial hierarchy, 
then a \#P-hard problem lies in the polynomial hierarchy.  

The {\em restricted membership problem}
for any probability vector $p^{n+2}$ is defined below:

\medskip

\fbox{\parbox{6in}{

\noindent \textbf{Name:} SVR$_{p^{n+2}}$-Membership Problem

\smallskip

\textbf{Input:} A vector $c=(c_1, \dots, c_{n+2}) \in \mathbb{Q}^{n+2}$.\\
\textbf{Question:} Is $c$ in $\mathcal{C}_{p^{n+2}}$ ? 

}}

\medskip

\begin{proof}[Proof of Theorem~\ref{thm:svr-hard}]
We first prove the following lemma that shows how an efficient oracle 
for the restricted verification problem 
can be used to obtain a membership oracle:

\begin{lemma}
If the SVR$_{p^{n+2}}$-Verification problem is in the $k$-th level of PH, 
then the SVR$_{p^{n+2}}$-Membership problem is in the $(k+1)$-level.
\end{lemma} 

\begin{proof}
Assume that the SVR$_{p^{n+2}}$-Verification problem is in the $k$-th level of PH. 
By Caratheodory's theorem a point $c$ is in $\mathcal{C}_{p^{n+2}}$ iff it is a convex combination 
of at most $n+3$ vertices of $\mathcal{C}_{p^{n+2}}$, i.e., $c=\sum_{i=1}^m \lambda_i x^{(i)}$, 
where $x^{(i)}$ is a vertex, $\lambda_i \geq 0$ for $1 \leq i \leq m \leq n+3$, 
and $\sum_{i=1}^m \lambda_i=1$. So, it can be certified that a given point $c$ is in 
$\mathcal{C}_{p^{n+2}}$ by finding the $m \leq n+3$ vertices $x^{(i)}$ and computing the $m$ 
scaling factors $\lambda_i$. Given the $x^{(i)}$, one can verify that $x^{(i)}$ is a vertex of $\mathcal{C}_{p^{n+2}}$ 
by finding a weight vector $w^{(i)}$ of the form described in Definition 4 such that $x^{(i)}$ 
is the $\widetilde{f}^{p^{n+2}}_{w^{(i)},0}$ vector, as the vertices of $\mathcal{C}_{p^{n+2}}$ 
correspond to linear threshold functions with weights of this specific form. So, if we are given 
the $x^{(i)}$ and the corresponding $w^{(i)}$, we can verify in polynomial time with a $k$-th level 
oracle that $x^{(i)}$ is the $\widetilde{f}^{p^{n+2}}_{w^{(i)},0}$ vector 
as we assumed that the SVR$_p^{n+2}$-Verification problem is in the $k$-th level of PH.
Thus, there is a polynomial-size certificate that can be checked in polynomial time with a $k$-th level oracle 
when $c$ is in $\mathcal{C}_{p^{n+2}}$: the $m \leq n+3$ vertices $x^{(i)}$, where each $x^{(i)}$ can be 
represented by poly($n$) bits by assumption; and the $m$ corresponding $w^{(i)}$ vectors, 
where each $w^{(i)}$ can be represented by poly($n$) bits, 
as every linear threshold function can be represented with weight $w=(w_1, \dots, w_{n+2})$ 
such that each $w_i$ is an integer that satisfies $|w_i| \leq 2^{(O(n\log n))}$ \cite{Mur61}. 
Given the $x^{(i)}$ and the $w^{(i)}$, it can be verified in polynomial time with a $k$-th level 
oracle that the $x^{(i)}$ are vertices and then we can compute in polynomial time the $\lambda_i$ 
coefficients by solving the linear system $c=\sum_{i=1}^m \lambda_i x^{(i)}$. Thus, if the 
SVR$_{p^{n+2}}$-Verification problem is in the $k$-th level of PH, 
the SVR$_{p^{n+2}}$-Membership problem is in the $(k+1)$-level. 
\end{proof}

As $\mathcal{C}_{p^{n+2}}$ has non-empty interior, if the SVR$_{p^{n+2}}$-Membership problem is in the $(k+1)$-level of PH, 
then using the ellipsoid algorithm, we could solve the optimization problem using a polynomial number 
of membership-oracle calls (page 189, \cite{ShL82}). Hence, we would have that the 
SV$_{p^{n+2}}$-Optimization problem, which by Theorem~\ref{thm:opt-hard} is $\#$P-hard, 
is in the $(k+2)$-level of PH. This completes the proof of Theorem~\ref{thm:svr-hard}.
\end{proof}

\paragraph{Hardness of the Verification Problem for Weighted Voting Games} 
One important issue is that the computational problems we have considered so far 
involve linear threshold functions some of whose weights can be negative. 
This seemed necessary to some extent for our arguments, as it is crucially exploited 
in the proof of Theorem~\ref{thm:opt-hard}.

We now show how to switch to weighted voting games (i.e., LTFs with non-negative weights), 
which was our initial goal. Using a bijection between the semivalues of a linear threshold function 
with weight vector $w = (w_1, \ldots, w_n)$ of the form described in Definition \ref{poly} 
and a linear threshold function with weight vector $w' = (|w_1|, \ldots, |w_n|)$, 
we show the equivalence between the restricted verification problem and the verification problem 
for linear threshold functions defined by positive weights.

\medskip

\fbox{\parbox{6in}{

\noindent \textbf{Name:} SV$_{p^{n}}$-Verification Problem 

\smallskip

\textbf{Input:} A vector $(c_1, \dots, c_{n}) \in \mathbb{Q}^{n}$, a vector $w=(w_1,\dots, w_{n}) \in \mathbb{Q}_{+}^n$ 
and $\theta \in \mathbb{Q}$.\\
\textbf{Question:} Does it hold that $\widetilde{f}^{p^{n}}_{w, \theta}(i)=c_i$ for $i \in \{1,\dots, {n}\}$? 

}}

\medskip

\begin{theorima} \label{thm:verification-equiv}
If the SV$_{p^{n+2}}$-Verification problem  is in the $k$-th level of PH, 
then the SVR$_{p^{n+2}}$-Verification problem is in the $k$-th level of PH.
\end{theorima}
\begin{proof}
We use the following lemma that shows how one can compute the semivalues 
of a linear threshold function with weight vector $w$ of the form 
described in Definition \ref{poly} given the semivalues of the linear threshold function 
with weight vector the absolute values of $w$, and vice-versa.

\begin{lemma} \label{pton}
For a positive integer $n$, fix $a_i > 0, 1 \leq i \leq n$, and $A=\sum_{i=1}^na_i$. 
Consider the LTFs $f(x)=\mathrm{sign}(\sum_{i=1}^n a_i x_i - (A/2) x_{n+1} -(A/2) x_{n+2})$ and
$g(x)=\mathrm{sign}(\sum_{i=1}^n a_i x_i+(A/2) x_{n+1}+(A/2) x_{n+2})$. Then, 
we have the following: 
\begin{enumerate}
\item[(i)] For $1 \leq i \leq n $, it holds
$\widetilde{g}^{p^{n+2}}(i)=\widetilde{f}^{p^{n+2}}(i)-2(p^{n+2}_{n+1}-p^{n+2}_{n-1})$, and 
\item[(ii)] For $n+1 \leq i \leq n+2$, it holds 
$\widetilde{g}^{p^{n+2}}(i)=\widetilde{f}^{p^{n+2}}(i)+2\sum_{t=0}^{n-1}{\binom{n}{t}}(p^{n+2}_t+p^{n+2}_{t+1})$.
\end{enumerate}
\end{lemma}
\begin{proof} 
For $i \in \{1, \dots, n\}$, from the definitions of semivalues and functions $f$, $g$ we get: 
\begin{dmath*}
\ \widetilde{g}^{p^{n+2}}(i)
=\sum_{x:x_i=-1}p_{wt(x)}^{n+2}(g(x^i)-g(x))
=\sum_{\substack{x:x_i=-1,\\x_{n+1}=x_{n+2}=-1}}p_{wt(x)}^{n+2}\left(\mathrm{sign}\left(\sum_{j\leq n:j \neq i}a_jx_j+a_i-A\right)-\mathrm{sign}\left(\sum_{j\leq n:j \neq i}a_jx_j-a_i-A\right)\right)+
\sum_{\substack{x:x_i=-1,\\x_{n+1}=x_{n+2}=1}}p_{wt(x)}^{n+2}\left(\mathrm{sign}\left(\sum_{j\leq n:j \neq i}a_jx_j+a_i+A\right)-\mathrm{sign}\left(\sum_{j\leq n:j \neq i}a_jx_j-a_i+A\right)\right)
\end{dmath*}
\begin{dmath*}
+\sum_{\substack{x:x_i=-1,\\x_{n+1}\neq x_{n+2}}}p_{wt(x)}^{n+2}\left(\mathrm{sign}\left(\sum_{j\leq n:j \neq i}a_jx_j+a_i\right)-\mathrm{sign}\left(\sum_{j\leq n:j \neq i}a_jx_j-a_i\right)\right)
\end{dmath*} 
and
\begin{dmath*}
\widetilde{f}^{p^{n+2}}(i)
=\sum_{x:x_i=-1}p_{wt(x)}^{n+2}(f(x^i)-f(x))
=\sum_{\substack{x:x_i=-1,\\x_{n+1}=x_{n+2}=1}}p^{n+2}_{wt(x)}\left(\mathrm{sign}\left(\sum_{j\leq n:j \neq i}a_jx_j+a_i-A\right)-\mathrm{sign}\left(\sum_{j\leq n:j \neq i}a_jx_j-a_i-A\right)\right)
+\sum_{\substack{x:x_i=-1,\\x_{n+1}=x_{n+2}=-1}}p^{n+2}_{wt(x)}\left(\mathrm{sign}\left(\sum_{j\leq n:j \neq i}a_jx_j+a_i+A\right)-\mathrm{sign}\left(\sum_{j\leq n:j \neq i}a_jx_j-a_i+A\right)\right)
+\sum_{\substack{x:x_i=-1,\\x_{n+1}\neq x_{n+2}}}p^{n+2}_{wt(x)}\left(\mathrm{sign}\left(\sum_{j\leq n:j \neq i}a_jx_j+a_i\right)-\mathrm{sign}\left(\sum_{j\leq n:j \neq i}a_jx_j-a_i\right)\right) \;.\end{dmath*} 
As $f(x)=-1$ for any $x \neq \bf{1}$ with $x_{n+1}=x_{n+2}=1$, 
$f(x)=1$ for any $x$ with $x_{n+1}=x_{n+2}=-1$, 
$g(x)=1$ for any $x$ with $x_{n+1}=x_{n+2}=1$, and 
$g(x)=-1$ for any $x\neq(1,\dots,1,-1,-1)$ with $x_{n+1}=x_{n+2}=-1$, 
we get that 
\[\widetilde{f}^{p^{n+2}}(i)-\widetilde{g}^{p^{n+2}}(i)=2(p^{n+2}_{n+1}-p^{n+2}_{n-1}) \;.\] 
For $i = n+1$, we can write:
\begin{dmath*} 
\widetilde{g}^{p^{n+2}}(n+1)
=\sum_{x:x_{n+1}=-1}p_{wt(x)}^{n+2}(g(x^{n+1})-g(x))
=\sum_{\substack{x:x_{n+1}=-1,\\x_{n+2}=-1}}p^{n+2}_{wt(x)}\left(\mathrm{sign}\left(\sum_{j\leq n}a_jx_j\right)-
\mathrm{sign}\left(\sum_{j\leq n}a_jx_j-A\right)\right)
+ \sum_{\substack{x:x_{n+1}=-1,\\x_{n+2}=1}}p^{n+2}_{wt(x)}\left(\mathrm{sign}\left(\sum_{j\leq n}a_jx_j+A\right)
-\mathrm{sign}\left(\sum_{j\leq n}a_jx_j\right)\right)\;. 
\end{dmath*}
and
\begin{dmath*} 
\widetilde{f}^{p^{n+2}}(n+1) 
=\sum_{x:x_{n+1}=-1}p_{wt(x)}^{n+2}(f(x^{n+1})-f(x))
= \sum_{\substack{x:x_{n+1}=-1,\\x_{n+2}=1}}p^{n+2}_{wt(x)}\left(\mathrm{sign}\left(\sum_{j\leq n}a_jx_j-A\right)
-\mathrm{sign}\left(\sum_{j\leq n}a_jx_j\right)\right) 
+\sum_{\substack{x:x_{n+1}=-1,\\x_{n+2}=-1}}p^{n+2}_{wt(x)}\left(\mathrm{sign}\left(\sum_{j\leq n}a_jx_j\right)
-\mathrm{sign}\left(\sum_{j\leq n}a_jx_j+A \right)\right) \;. 
\end{dmath*} 
We thus have 
\begin{dmath*} 
\widetilde{f}^{p^{n+2}}(n+1)-\widetilde{g}^{p^{n+2}}(n+1)
=\sum_{\substack{x:x_{n+1}=-1,\\x_{n+2}=-1}}p^{n+2}_{wt(x)}\left(\mathrm{sign}\left(\sum_{j\leq n}a_jx_j-A\right)
-\mathrm{sign}\left(\sum_{j\leq n}a_jx_j+A\right)\right)
+\sum_{\substack{x:x_{n+1}=-1,\\x_{n+2}=1}}p^{n+2}_{wt(x)}\left(\mathrm{sign}\left(\sum_{j\leq n}a_jx_j-A\right)
-\mathrm{sign}\left(\sum_{j\leq n}a_jx_j+A\right)\right)
=-2\sum_{t=0}^{n-1}{n \choose t}(p^{n+2}_t+p^{n+2}_{t+1}) \;. 
\end{dmath*} 
In the same way, we get that 
\[\widetilde{f}^{p^{n+2}}(n+2)-\widetilde{g}^{p^{n+2}}(n+2)
=-2\sum_{t=0}^{n-1}{n \choose t}(p^{n+2}_t+p^{n+2}_{t+1}) \;.
\] 
This completes the proof of Lemma~\ref{pton}. 
\end{proof}

Given an instance $(a_1, \dots, a_n, -A/2, -A/2)$, $(c_1, \dots, c_{n+2})$ of the SVR$_{p^{n+2}}$-Verification problem, 
we construct the following instance of SV$_{p^{n+2}}$-Verification problem:  
$(a_1, \dots, a_n, +A/2, +A/2)$, $\theta=0$, 
\begin{align*}
&\Big(c_1-2(p^{n+2}_{n+1}-p^{n+2}_{n-1}), \dots, c_n-2(p^{n+2}_{n+1}-p^{n+2}_{n-1}), c_{n+1}+2\sum_{t=0}^{n-1}{\binom{n}{t}}(p^{n+2}_t+p^{n+2}_{t+1}),\\
&c_{n+2}+2\sum_{t=0}^{n-1}{\binom{n}{t}}(p^{n+2}_t+p^{n+2}_{t+1})\Big) \; .
\end{align*}
By Lemma \ref{pton}, we have that the SVR$_{p^{n+2}}$-Verification instance is a ``YES''-instance 
iff the SV$_{p^{n+2}}$-Verification instance is a ``YES''-instance. 
This completes the proof of Theorem~\ref{thm:verification-equiv}.
\end{proof}

\paragraph{Verification Using Inverse Oracle} 
The final step of our proof is to show that the inverse problem for semivalues 
is at least as hard as the verification problem. While this is intuitively obvious,
the proof requires the following non-trivial structural result:
The semivalues of a weighted voting game characterize the game within the space of weighted voting games. 

\begin{theorima} \label{thm:inverse-vs-ver}
If the SV$_{p^{n}}$-Inverse problem is in the $k$-th level of PH, 
then the SV$_{p^{n}}$-Verification problem is in the $(k+1)$-level.
\end{theorima}
\begin{proof}
The proof makes essential use of the following lemma that shows 
that if two LTFs with normalized weights and 
the same threshold have the same semivalues, then they are equal
on all points $x$ that are given positive probability by the distribution $\mu'_{p^n}$. 
This lemma is qualitatively similar to (and inspired by) Chow's Theorem~\cite{chow}, 
that shows that a linear threshold function is uniquely determined by its 
Chow parameters:

\begin{lemma} \label{equiv}
Let $f(x) = \mathrm{sign}(w \cdot x - \theta)$ and $g(x) = \mathrm{sign}(v \cdot x - \theta)$
where $\sum_{i=1}^n w_i = \sum_{i=1}^n v_i$. If $\widetilde{f}^{p^{n}}(i)=\widetilde{g}^{p^{n}}(i)$ for $i \in \{1, \dots, n\}$, 
then $f(x)=g(x)$ for all $x \in \{-1,1\}^n$ such that $p^n_{wt(x)}+p^n_{wt(x)-1} \neq 0$ 
and $|w \cdot x-\theta|+|v \cdot x-\theta| \neq 0$.
\end{lemma}
\begin{proof} 
By our assumption that $\widetilde{f}^{p^{n}}(i)=\widetilde{g}^{p^{n}}(i)$ for $i \in \{1, \dots, n\}$
it follows that:
\begin{align*}
&2\sum_{i=1}^n w_i (\widetilde{f}^{p^n}(i)-\widetilde{g}^{p^{n}}(i))+
2\sum_{i=1}^n v_i(\widetilde{g}^{p^{n}}(i)-\widetilde{f}^{p^n}(i))
-\theta \sum_{x \in \{-1,1\}^n}(f(x)-g(x))\mu'_{p^n}(x) \\&
- \theta\sum_{x \in \{-1,1\}^n}(g(x)-f(x))\mu'_{p^n}(x)=0
\end{align*}
Recalling our reformulation of the semivalues (\ref{sem}), we equivalently have:
\begin{align*}
&\sum_{i=1}^n w_i \left(\widehat{f}^{p^n}(i)-\widehat{g}^{p^n}(i)+C_f^{p^n}-C^{p^n}_g\right)
+ \sum_{i=1}^n v_i \left(\widehat{g}^{p^n}(i)-\widehat{f}^{p^n}(i)+C_g^{p^n}-C^{p^n}_f\right)\\
&-\theta \sum_{x \in \{-1,1\}^n}(f(x)-g(x)) \mu'_{p^n}(x) 
- \theta \sum_{x \in \{-1,1\}^n}(g(x)-f(x)) \mu'_{p^n}(x)=0\\
& \Leftrightarrow 
\sum_{x \in \{-1,1\}^n}(f(x)-g(x))\mu'_{p^n}(x) (w \cdot x) 
+\sum_{x \in \{-1,1\}^n}(g(x)-f(x))\mu'_{p^n}(x) (v \cdot x) \\
&-\theta\sum_{x \in \{-1,1\}^n}(f(x)-g(x))\mu'_{p^n}(x)
-\theta\sum_{x \in \{-1,1\}^n}(g(x)-f(x))\mu'_{p^n}(x)=0 \\
& \Leftrightarrow \sum_{x \in \{-1,1\}^n} (f(x)-g(x)) \mu'_{p^n}(x) (w\cdot x-\theta)
+\sum_{x \in \{-1,1\}^n}(g(x)-f(x))\mu'_{p^n}(x)(v \cdot x-\theta)=0\\
&\Leftrightarrow \sum_{x \in \{-1,1\}^n} \mu'_{p^n}(x)|f(x)-g(x)| |w \cdot x-\theta|
+\sum_{x \in \{-1,1\}^n}\mu'_{p^n}(x)|g(x)-f(x)| |v \cdot x-\theta|=0 \\
& \Leftrightarrow \sum_{x \in \{-1,1\}^n} \mu'_{p^n}(x) |f(x)-g(x)| (|w \cdot x-\theta|+|v \cdot x-\theta|)=0 \;.
\end{align*}
Hence, for any $x \in \{-1,1\}^n$ such that $|w \cdot x-\theta|+|v \cdot x-\theta| \neq 0$ 
and $\mu'_{p^{n}}(x) \neq 0$, we have that $f(x)=g(x)$.
This completes the proof of Lemma~\ref{equiv}.
\end{proof}

We are now ready to complete the proof.
Given an SV$_{p^{n}}$-Verification instance
$a=(a_1, \dots, a_n)$, $\theta$, $(c_1, \dots, c_{n})$, 
we create the following instance of the SV$_{p^n}$-Inverse problem: 
$\theta' = \theta/\sum_{i=1}^n{a_i}$, $(c_1, \dots, c_{n})$. Then, if the SV$_{p^{n}}$-Inverse instance 
is a ``NO''-instance, we have a ``NO''-instance of the SV$_{p^{n}}$-Verification problem. 
If the SV$_{p^n}$-Inverse output is a weight vector $w=(w_1, \dots, w_n)$, 
we can check with a co-NP oracle if the functions $f_{w,\theta'}$ 
and $f_{a,\theta}$ have the same semivalues: 
By Lemma \ref{equiv}, they have the same semivalues iff 
there is no $x \in \{-1, 1\}^n$ such that $\mu'_{p^n}(x) > 0$ and  
$f_{w,\theta'}(x) \neq f_{a,\theta}(x)=f_{\frac{a}{\sum_{i=1}^n{a_i}}, \frac{\theta}{\sum_{i=1}^n{a_i}}}(x)$.
This completes the proof of Theorem~\ref{thm:inverse-vs-ver}.
\end{proof}

Theorem~\ref{main} now follows by combining Theorems~\ref{thm:svr-hard},~\ref{thm:verification-equiv}, 
and~\ref{thm:inverse-vs-ver}.

\section{Conclusions}
The inverse power index problem has received considerable attention in game theory and social choice, 
and the inverse Banzhaf index problem has been relevant in other fields as well, such as circuit complexity 
and computational learning. In this paper, we proved that the inverse semivalue problem, for {\em reasonable} 
probability distributions, is computationally intractable. As special cases, we deduce that the inverse Banzhaf 
index and inverse Shapley value problems are also intractable. 
A number of interesting open questions remain: Can we design efficient approximation algorithms for the inverse problem 
in the case of more general semivalues? Can we characterize the computational complexity of the inverse power index problem 
for power indices that do not belong in the semivalues class?

\paragraph{Acknowledgements.} We thank Shaddin Dughmi for his contributions to the early stages of this work.
We are grateful to Kousha Etessami for numerous helpful discussions.

\bibliographystyle{alpha}

\bibliography{biblio_update}

\end{document}